\documentclass[a4paper,british]{article}

\usepackage{amsmath,amssymb,amsthm,bm,color,babel}
\usepackage[bookmarks,colorlinks,breaklinks]{hyperref}  

\renewcommand{\Phi}{\sigma}

\newtheorem{thm}{Theorem}[section] 
\newtheorem{prop}[thm]{Proposition}
\newtheorem{Def}[thm]{Definition}

\newcommand{\Tr}{{\rm Tr}}

\title{Quantum Geometry on Quantum Spacetime: 
Distance, Area and Volume Operators}
\author{D.~Bahns\thanks{Courant Research Center Mathematics, 
University of G\"ottingen,~--~Bunsenstr. 3-5, D - 37073 G\"ottingen. Email: 
{\tt bahns@uni-math.gwdg.de}~.},
S.~Doplicher\thanks{Dipartimento di Matematica, Universit\`a degli Studi di Roma ``Sapeinza'', P.le A.~Moro 2, 00185 Roma. Email: 
{\tt dopliche@mat.uniroma1.it}~.}, 
K.~Fredenhagen\thanks{II. Institut f\"ur Theoretische Physik, Universit\"at 
Hamburg, Luruper Chaussee 149, D-22761, Hamburg. E-mail: 
{\tt klaus.fredenhagen@desy.de}~.} and 
G.~Piacitelli\thanks{SISSA~--~Via Bonomea 256, 34136 Trieste, Italy. E-mail: 
{\tt gherardo@piacitelli.org}~.}}

\date{}
\begin{document}

\maketitle

\begin{abstract} 
We develop the first steps towards an analysis of geometry on the
quantum spacetime proposed in [1]. The homogeneous elements of the
universal differential algebra are naturally identified with operators
living in tensor powers of Quantum Spacetime;  this allows 
us to compute
their spectra. In particular, we consider operators that can be
interpreted as distances, areas, 3- and 4-volumes.

The Minkowski distance operator between two independent events is shown to have pure Lebesgue spectrum with infinite multiplicity.
The Euclidean distance operator is shown to have spectrum bounded below by 
a constant of the order of the Planck length. The corresponding 
statement is proved also for both the space-space and space-time area 
operators, as well as for the Euclidean length of the vector representing the 
3-volume operators. However, the space 3-volume operator (the time component 
of that vector) is shown to have spectrum equal to the whole complex plane. All 
these operators are normal, while the distance operators are also 
selfadjoint.

The Lorentz invariant spacetime volume operator, representing the 4-volume 
spanned by five independent events, is shown to be normal. Its spectrum is pure point with a 
finite distance (of the order of the fourth power of the Planck length) 
away from the origin. 

The mathematical formalism apt to these problems is developed and its 
relation to a general formulation of Gauge Theories on Quantum Spaces is 
outlined. As a byproduct, a Hodge Duality between the absolute 
differential and the Hochschild boundary is pointed out.
\end{abstract}

\newpage

\section{Introduction and preliminaries}
The concurrence of the principles of Quantum
Mechanics and of Classical General Relativity imposes limits on the joint precision allowed in 
the measurement of the four spacetime coordinates of an event, as a consequence of the 
following principle:
\begin{quote}
{\it The gravitational field generated by the concentration of energy  
required by the Heisenberg Un\-cer\-tainty Principle to localise an event
in
spacetime should not be so strong to hide the event itself to any distant
observer - distant compared to the Planck scale.}
\end{quote}
These limitations pose no restriction on the precision in the measurement 
of a 
{\itshape single} coordinate\footnote{This does not conflict with the famous 
Amati Ciafaloni Veneziano {\itshape Generalised Uncertainty 
Relation}: all the derivations  we are aware of (see e.g.\ \cite{Amati:1990xe,Mead:1964zz})
implicitly assume that {\itshape all} space 
coordinates of the event are measured with uncertainties of the same order of 
magnitude; in which case they agree with our Spacetime Uncertainty Relations.}, 
but lead to {\itshape Spacetime Uncertainty 
Relations} 
\begin{equation}
\label{dopleq1}
\Delta q^0 \cdot \sum \limits_{j = 1}^3 \Delta q^j \gtrsim  \lambda_P ^{2} ; 
\quad
\sum \limits_{1 \leq j <  k \leq 3 } \Delta q^j  \Delta q^k \gtrsim \lambda_P ^{2},
\end{equation}  
that were shown to be exactly implemented by Commutation Relations 
between coordinates, which turn spacetime into {\itshape Quantum Spacetime} 
\cite{Doplicher:1994tu, Doplicher:1994zv}.

More precisely, the four spacetime coordinates of an event are described 
by four operators which fulfil
\begin{equation}
\label{dopleq2}
[q^\mu  ,q^\nu  ]   =   i \lambda_P^2   Q^{\mu \nu}
\end{equation}
where $\lambda_P$ denotes the {\itshape Planck length}, 
\[
\lambda_P=\left({G\hbar\over
c^3}\right)^{1/2}\simeq1.6\times10^{-33}\text{ cm,}
\]
hereafter set 
equal to $1$ adopting  {\itshape absolute units},  where \( \hbar = c =
G  = 1\); $Q$ fulfills 
the {\itshape Quantum Conditions}
\begin{subequations}
\label{quantum_equations}
\begin{gather}
\label{dopleq3}
\left(1/4)[q^0  ,q^1,q^2,q^3 \right]^2  =  I,\\
 [q^\mu  ,q^\nu  ]   [q_{\mu}  ,q_{\nu}] = 0, \\
[[q^\mu  ,q^\nu  ], q^\lambda] = 0.
\end{gather}
\end{subequations}
Namely, the only full Lorentz invariant constructed with the commutator 
which is required to be nonzero is the square of the pseudoscalar
\begin{subequations}
\begin{align}
\nonumber
\left[q^0  ,q^1,q^2,q^3 \right]  &\equiv  \det \left(
\begin{array}{cccc}
q^0 & q^1&q^2 & q^3 \\
q^0 & q^1&q^2 & q^3 \\
q^0 & q^1&q^2 & q^3 \\
q^0 & q^1&q^2 & q^3 
\end{array}
\right)\\
\nonumber
&\equiv  \varepsilon_{\mu \nu \lambda \rho} q^\mu q^\nu q^\lambda
q^\rho =\\
&= - (1/2) Q^{\mu \nu}  (*Q)_{\mu \nu},
\label{dopleq4}
\end{align}
where 
\begin{equation}\label{eq:hodge_of_Q}
(*Q)_{\mu\nu}=\tfrac 12\varepsilon_{\mu \nu \lambda \rho}Q^{\lambda\rho}.
\end{equation}
\end{subequations}

In this model, the {\itshape C*-Algebra of Quantum Spacetime} $ \mathcal 
E$ is defined as the enveloping C*-Algebra of the Weyl form (selecting 
the {\itshape regular representations}\footnote{An irreducible 
representation of (1) is regular if and only if it generates an 
irreducible representation of $ \mathcal
E$, or, by the Dixmier - Nelson Theorem, if and only if the sum of the 
squares of the $q^\mu$ is essentially selfadjoint. In this paper we will 
deal only 
with such representations.}) of the commutation relations between 
the coordinates: 
\[
 e^{i\alpha_\mu q^\mu}e^{i\beta_\mu q^\nu}=e^{-(i/2)\alpha_\mu 
Q^{\mu\nu}\beta_\nu} 
e^{i(\alpha+\beta)_\mu q^\mu}\ ;\quad\alpha,\beta\in\mathbb R^4. 
\]

The unbounded operators \(q^\mu\) are {\itshape affiliated} to the 
C*-Algebra \(\mathcal E\) and fulfill the desired commutation relations. 
Poincar\'e covariance is expressed by an action $\tau$ of the full 
Poincar\'e group by automorphisms of $ \mathcal E$, determined by the 
property that its canonical extension to the $ q $ 's fulfil
\[
\tau_L ( q)  =  L^{-1} (q).
\]
The C*-Algebra $ \mathcal E$ turns out to be the C*-Algebra of continuous
functions vanishing at infinity from a manifold $\Sigma$ to the 
C*-Algebra of compact operators on the separable infinite dimensional Hilbert
space.

Here $\Sigma$ is the (maximal) {\itshape joint spectrum of the
commutators}, which is
the manifold of the real antisymmetric two-tensors fulfilling the
constraints imposed by the above {\itshape quantum conditions}; namely,  
specifying such a tensor by its electric and magnetic components $\vec e,
\vec m$, $\vec e{\,}^2=\vec m^2$, $\vec e\cdot\vec m= \pm 1$. Thus $\Sigma$
can be identified with the full Lorentz orbit of the standard symplectic 
form in four dimensions, that is $\Sigma$ is the union of two connected  
components, each homeomorphic to $SL(2,\mathbb C)/ {\mathbb C}_* $, or to
the tangent manifold $TS^2$ to the unit sphere in three dimensions. If 
$\vec e = \pm \vec m$ they must be of length one, and span the {\itshape
base} $\Sigma ^{(1)} $ of $\Sigma$. Thus $\Sigma$ can be viewed as $ T
\Sigma ^{(1)} $.

If we choose
$\sigma=(\vec{e},\vec{m})$ in $\Sigma ^{(1)} $  with $\vec{e}=\vec{m}=(1,0,0)$, the 
corresponding irreducible representation of the operators
$q^{\mu}$ can be realised on the Hilbert space $H \otimes H$, with  $ H = 
L^2(\mathbb{R})$ by 
\begin{subequations}
\label{basic_rep}
\begin{eqnarray} 
q^0 = Q \otimes I,\\ 
q^1 = P \otimes I,\\
q^2 = I \otimes Q,\\
q^3 = I \otimes P,
\end{eqnarray} 
\end{subequations}
where $Q$ is the operator of
multiplication by $s$ in $H = L^2(\mathbb{R}, ds)$ and $P$ is the operator 
$\frac{1}{i}\frac{\partial}{\partial s}$ on $H$.

To begin with, we examine the Minkowski square operator $q^{\mu}q_{\mu} $ in Quantum 
spacetime. Since it is a full Lorentz scalar, it can be computed at a fixed point $\sigma$ in 
$\Sigma$. Hence, choosing $\sigma$ as here above, we can compute $q^{\mu}q_{\mu} $ using that 
particular irreducible representation as
\[
q^{\mu}q_{\mu} = (Q^2 - P^2) \otimes I - I \otimes ((Q^2 + P^2).
\]
While $Q^2 + P^2$ has pure point spectrum, equal to ${2n + 1, n = 0, 1, 
2,\dots}$, $(1/2)(P^2 - Q^2)$ is the Hamiltonian of the anti-harmonic 
oscillator, which is known to be unitarily equivalent to the direct sum of 
two copies of $Q$ (or of $P$). For instance, a unitary $U$ changing $P 
\oplus 
P $ to 
$(1/2) (Q^2 - P^2)$ by its adjoint action can be given as
\[
(U(\psi_1 \oplus \psi_2)(s) = e^{(i/2)(-i \frac{\partial}{\partial s})^2 } 
e^{(i/4)s^2 } 
(1/{|s|^{1/2}})(\theta(s) \psi_1 (ln |s|) + \theta(- s) \psi_2 (ln |s|)).
\]  
Now $ a P + b I $ is unitarily equivalent to $P$ and to $Q$ for any real $a, b$, with $a$ 
nonzero;  
therefore $q^{\mu}q_{\mu} $ is unitarily equivalent to the direct sum of infinitely many copies of 
$Q$, in any irreducible representation, hence in any representation of  $ \mathcal E $.

The Euclidean square operator $\Sigma q_{\mu}^2$ has been analysed in \cite{Doplicher:1994tu}. In the above 
representation, it agrees with the operator,
\[
 (P^2 + Q^2) \otimes I + I \otimes (P^2 + Q^2),
\]
with pure point spectrum whose minimum is $2$, corresponding to a state optimally localised 
around the origin. It was shown there that this is the minimum of the spectrum over all 
representations, and can be attained only in pure states in irreducible representations 
associated to points in the base  $\Sigma ^{(1)} $ of $\Sigma$, or in states obtained from them 
by integration with a regular probability measure on $\Sigma ^{(1)}$.

These statements will extend immediately to the Minkowskian 
respectively Euclidean distance operators between two independent events.

In Quantum Mechanics the observables for the system composed of independent 
subsystems are described by tensor products of observables for the subsystems.
This suggests to describe the algebra of coordinates for $n$ independent 
spacetime events as the tensor product $ \mathcal E \otimes \mathcal E \otimes 
\dots \otimes \mathcal E$ ($n$ times). Thus the coordinate operators for the 
$j$-th event are 
\[
q_j  = I \otimes I \otimes\dots I \otimes q \otimes I \otimes\dots \otimes I,
(j\text{-th place}).
\]

It is however more natural to use not the standard tensor product over the 
complex numbers, but the {\itshape ${\mathcal Z}$-module tensor product}, 
where   $ {\mathcal Z} $ denotes the centre of the multiplier algebra of 
$\mathcal E$, that is the algebra of continuous bounded complex functions on $ 
\Sigma $. This amounts to require that  the commutators of the different spacetime components of $ q_j $ are independent of $j$, so that
\[
[q_{j}^{\mu} , q_{k}^{\nu}] = i Q^{\mu\nu} \delta _{j,k},
\] 
or, in the language of the next section, that
\[
d Q = 0.
\]
This choice has been discussed and motivated in \cite{Bahns:2003vb}, and implies in 
particular that the normalised difference variables
\[
\frac{1}{\sqrt{2}}(q_j - q_k )
\]
for distinct $j,k$ obey {\itshape the same} commutation relations \eqref{dopleq1}, and 
commute with the barycentre coordinates.

In particular, the  Minkowski, respectively Euclidean, distance operators, 
evaluated with these normalised difference variables for two independent 
events, will have exactly the same spectral properties as above described for 
the $q$s themselves, so that, in absolute units,
\[
\sum \limits_{\mu = 0}^3 (q_{j}^{\mu} - q_{k}^{\mu})^{2} \gtrsim 4.
\]

In the next section we will briefly discuss a variant of the Universal 
Differential Calculus where this formalism fits best, and the two distinct 
relevant algebraic structures. This is the natural mathematical ground for 
a general formulation of Gauge Theories on Quantum Spaces, briefly touched 
upon in the subsequent Section.

In the course of the discussion we will present a novel pairing between 
differential forms, the {\itshape q-pairing}, which yields a {\itshape 
duality} between the exterior differential and the Hochschild boundary, which 
makes of the latter a {\itshape codifferential}.

In this formalism, the exterior products of the $dq$'s will describe precisely 
the operators for the area, volume, four-volume spanned by respectively 3, 
4 or 5 
points describing independent events. 

The spectrum of each of these operators, however, has to be defined 
relative to another structure of *-Algebra, namely viewing the 
differential forms 
as operators affiliated to the
C*-Algebra direct sum of the C*-algebraic $\mathcal Z$-module tensor 
powers of $\mathcal E$.

A similar distinction of structures underlies the usual differential 
calculus.

The spectra of these operators will be fully computed in the last Section. The 
results confirm what might be anticipated on intuitive grounds, based on the 
uncertainty relations: if one space coordinate has a very small uncertainty 
$a$, then at least one other space coordinate and the time coordinate must have
an uncertainty $b$ such that $a b  \gtrsim 1$; this suggests that the space - 
space and the space - time areas ought to be bounded below by $1$, as well as 
the spacetime volume; while the space 3-volume (the time component of 
the vector representing the 3-volume operators) might be arbitrarily small, 
since,
in the above situation, the third space coordinate might well have an 
uncertainty of the same order $a$, so that the product of the three space 
uncertainties has the same arbitrarily small order; while the further product 
with the  time uncertainty, of order $b$, is again at least of order one. 
These 
arguments could be casted into a mathematical proof that the product of the 
four
spacetime uncertainties must be bounded below by order 
one as a consequence of (1).

It is worth noting that the restrictions we find here upon the spectra of the 
distance, area, and volume operators are just the {\itshape minimal} 
restrictions, 
imposed merely by the underlying geometry of our model of quantum Spacetime. 
In a realistic theory, the restrictions coming from dynamics together with the 
expected interconnections of spacetime and fields \cite{Doplicher:2001qt} might well 
impose tighter limitations.

\section{Independent events and the universal differential calculus}

Let $A$ be an associative algebra with unit over $\mathbb{C}$, obtained e.g. by adding the unit to an algebra $A_0$. Any tensor power $\Lambda^n (A):=A^{\otimes (n+1)}$ of $A$ over $\mathbb{C}$ can be viewed as an $A$-bimodule (using the 
product in $A$ on the first, resp. last factor), and the direct sum
\[ 
\Lambda (A)  = \bigoplus\limits_{n = 0}^\infty \Lambda _n (A) 
\] 
can be viewed as the $A$-bimodule tensor algebra\footnote{
For closely related notions, as 
the free product of two algebras, see \cite{Cuntz_martina_franca}.}, 
so that, in $\Lambda (A) $, 
\begin{equation}
(a_0 \otimes \dots \otimes a_n ) \cdot (b_0 \otimes \dots \otimes b_m) = a_0 
\otimes \dots \otimes a_n b_0 \otimes \dots \otimes b_m. 
\end{equation}
The exterior differential is defined on $\Lambda 
(A) $ by 
\begin{equation}
d (a_0 \otimes \dots \otimes a_n)  = \sum \limits_{k = 0}^{n+1} (-1)^k a_0 
\otimes \dots \otimes a_{k-1} \otimes I \otimes a_k\otimes \dots \otimes a_n . 
\end{equation}
As usual $d$ is a graded differential, i.e., if 
$\phi \in \Lambda^n (A), \psi \in \Lambda (A)$, we have 
\begin{align*} 
d^2 &= 0;\\
d (\phi\psi ) &= (d\phi )\psi + (-1)^n \phi d \psi. 
\end{align*} Note that $A = \Lambda^0 (A) 
\subset\Lambda (A)$, and the $d$-stable subalgebra $\Omega (A)$ of $\Lambda (A)$ generated by $A$ is the { \itshape universal differential 
algebra} (or the universal differential algebra of $A_0$ following 
\cite[{\S}1.\(\alpha\)]{connes}).

Observe that $\Omega (A)$ defined as above coincides with the usual notion of the universal differential algebra $\bigoplus_{n\geq 0} \Omega^n(A)$
with $\Omega^0(A):=A$, and 
\[
\Omega^n(A):=\Omega^1(A)\otimes_A \cdots \otimes_A \Omega^1(A) \qquad \mbox{($n$ times)}
\]

For $n=1,2,\dots$, consider the multiplication maps 
$m_{k}:\Lambda^{n+1}( A)\to\Lambda^{n}( A)$  
which linearly extend 
\[
m_k ( a_0 \otimes \dots \otimes a_n ) :=  a_0 \otimes \dots \otimes a_k a_{k+1} \otimes \dots \otimes a_n,
 \] 
where $k=0,\dots, n-1$. Then  $\Omega^n(A)$ coincides with the intersection of 
the kernels of all multiplication maps \cite{dubois-violette}.

On $\Lambda (A) $, we define an $A$-valued pairing, called the  {\itshape 
q-pairing}, by linearly extending
\[ 
\langle a_0 \otimes \dots \otimes a_n , b_0 \otimes \dots \otimes b_m\rangle  
:= \delta_ {n,m} \, a_0 b_0 \dots a_n b_n.
\] 

Observe that for any $a,b \in A$, $v \in \Lambda^n(A)$, $w \in \Lambda^m(A)$,
we have
\begin{equation}\label{eq:pairingMod}
\langle av ,  wb \rangle = a \, \langle v , w \rangle \, b 
\end{equation}

Moreover,
\begin{gather}
\langle da\,\omega,db\,\phi\rangle =[a,b]\langle\omega,\phi\rangle\\
\langle \omega\,da,\phi\,db\rangle =\langle\omega,\phi\rangle[a,b]
\end{gather}
for all  \(a,b\in A\), \(\omega,\phi\in\Lambda^n(A)\). It follows that,
when restricted to $\Omega(A)$, the $q$-pairing has some additional properties:
for any 
$a_i, b_i, a, b \in A$, and 
$\omega, d\psi \in \Omega^n(A)$, and $\phi , \lambda \in \Omega^m(A)$, we have
\begin{subequations}\label{eq:add-properties}
\begin{eqnarray}
\langle d a_1... da_n , d b_1 ... d b_n \rangle  &=& 
[a_1 , b_1 ] ... [a_n , b_n ] 
\\
\langle  \omega  \phi , (d \psi ) \lambda \rangle  &=& \langle  \omega  , d \psi \rangle  \langle  \phi , \lambda \rangle 
\\
\langle  \lambda d \psi ,  \phi \omega  \rangle & = & \langle  \lambda, \phi \rangle   \langle  d \psi , \omega \rangle 
\end{eqnarray}
\end{subequations}

The $A$-valued pairing $\langle \cdot, \cdot \rangle$ turns into an interesting  $\mathbb{C}$ - valued 
pairing by composition with a trace $\tau$. Namely, let $\tau$ be a complex valued 
linear map defined on a two sided ideal $J$, such that
\[ 
\tau (ab) = \tau (ba),\quad a \in A, b \in J;
\]
$\tau$ is {\itshape faithful} if $a$ in $A$ fulfills $\tau (ab) = 0$ for all $b$ in 
$J$ only if $a = 0$.

If $\Lambda _J (A)$ denotes the {\itshape differential ideal} in $\Lambda (A)$ 
generated by $J$, we have that $\Lambda _J (A)$ is the span of elements $ a_0 
\otimes \dots \otimes a_n $ where $a_j \in J$ for at least one $j$, and
\[
\langle \phi, \psi\rangle  \in J,\quad \phi \in \Lambda (A), \psi \in \Lambda _J (A).
\] 

Let $\delta$ denote the {\itshape Hochschild boundary} defined by 
\begin{align*}
\delta (a_0 \otimes \dots \otimes a_n)  
=& \sum \limits_{k = 0}^{n-1} (-1)^k a_0 
\otimes
\dots \otimes a_ {k-1} \otimes a_ {k}a_ {k+1} \otimes   a_{k+2} \otimes \dots \otimes 
a_n  +\\&+ (-1)^n a_n a_0 \otimes \dots \otimes a_{n-1}.
\end{align*}
Then we have

\begin{prop}
      The Hochschild boundary is a Hodge dual of the differential for the 
pairing $\tau ( \langle  \cdot , \cdot \rangle )$, namely
\[
\tau (\langle \delta \omega , \phi\rangle ) = \tau (\langle \omega , d \phi\rangle ), \omega , \phi \in \Lambda 
(A).
\]
\end{prop}
\begin{proof}
A routine computation.
\end{proof}

Thus the associated Laplacian is $ ( d + \delta )^2 = d \delta + \delta d $; this 
operator, and the associated Hodge theory, has been studied by Cuntz and Quillen in 
the context of  Cyclic cohomology \cite{Cuntz_martina_franca}.

A similar conclusion as above might be drawn using, in place of a trace $\tau$, 
the 
{\itshape universal trace}, namely the quotient map of $A$ modulo the vector 
subspace spanned by the commutators \cite{cuntz_private}.

The pairing discussed here reduces exactly to zero in case of a commutative algebra, 
so it seems to describe an alternative to the classical formalism, possibly valid 
only in the purely quantum picture, rather than a deformation of the classical 
calculus, which reduces to it in the limit where the relevant parameter, as the 
Planck length, is neglected. This property is however fulfilled, in the case of the 
model of 
Quantum Spacetime described in the Introduction, by another pairing,
that we will mention here below.

Note first that the above formalism can be applied to $\mathcal E$, or to 
$\mathcal E$ embedded in its multiplier algebra with centre $\mathcal Z$; 
we will actually use it also for the affiliated unbounded operators as 
$q$. But it will be important for our application to use everywhere in the 
above formalism the $\mathcal Z$ -{\itshape module tensor product}, as 
anticipated in the introduction, that is \( dQ = 0\).
 We will still write $\Lambda _n (M {\mathcal E})$ for the $n$th 
$\mathcal Z$ - module tensor product.

In the language of the universal differential calculus, the difference operator for coordinates of two independent events as discussed in the previous section, i.e.
$
q_2-q_1= 1\otimes q - q\otimes 1 
$
is given as $dq$. Observe that this order 
makes sense since on a (commutative) vector space, \(q_2-q_1\) is 
the vector that connects \(q_1\) with \(q_2\), so it has the same sign as the 
tangent vector of a straight curve from \(q_1\) to \(q_2\).
Furthermore, the geometric operators describing area, volume and spacetime
volume on Quantum Spacetime spanned by the independent points will be
described by the exterior products of the $dq$'s, evaluated in $\Lambda 
(M{\mathcal E})$.

In this setting, we can introduce other pairings $( \cdot, \cdot)$ on $\Omega(\mathcal E)$ by choosing suitable elements $\mu$ in $\Lambda (\mathcal E)$, and setting 
\begin{equation}\label{eq:new-pairing}
( \phi , \psi )  = \tau ( \langle \phi \psi, \mu \rangle ) \delta _{n,m} ,
\quad\phi \in \Omega _n (\mathcal E) , \psi \in \Omega _m (\mathcal E) \ .
\end{equation}
Using the product of differential forms in $\Lambda(\mathcal E)$, we now consider a particular choice of $\mu$ which has a natural interpretation as a metric form. To this end, first observe that the $\nu$-component of spacetime derivations on Quantum Spacetime is given by $\text{iAd}\, \tilde q_\nu$, where $\tilde q_\nu= (Q^{-1} q)_{\nu}$. Now consider 
\begin{equation}\label{def:mu-metric}
\mu := \sum_n (d{\tilde q^{[\mu_1}}  d{\tilde q^{\mu_2}}\dots  d{\tilde q^{\mu_n]}}) (d{\tilde q_{[\mu_1}} d{\tilde q_{\mu_2}}\dots  d{\tilde q_{\mu_n]}}) 
\end{equation}
with totally antisymmetrized products in each of the brackets and Lorentz contraction of upper and lower indices. As we shall see in section 4, this metric operator, 
$\mu=(d\tilde q\wedge\dotsm\wedge d\tilde q)
(d\tilde q\wedge\dotsm\wedge d\tilde q)$ in the notation 
adopted there, can be understood in terms of distance, area and 
volume operators.

We then find for the corresponding pairing
\begin{equation}
(a_0 da_1 \dots da_n, db_1 \dots db_n b_{n+1}) = \tau (a_0 \partial _{[ \mu _1 
} a_1 \dots \partial _{ \mu _n ] } a_n \partial ^{[ \mu _1 } 
b_1 \dots \partial ^{ \mu _n ] } b_n b_{n+1}) , 
\end{equation}
which justifies the interpretation of $\mu$ as a metric form.

\section{Connection  and parallel transport}

In order to introduce gauge theories on noncommutative spaces, one usually 
starts from modules, interpreted as spaces of sections of a vector bundle with 
a noncommutative base. Connections for the universal differential calculus can 
be introduced as covariant derivatives, 
provided the module is 
projective\footnote{We are grateful to Jochen Zahn for drawing our attention to \cite[Prop 8.3]{GraciaBondia:2001tr}.}. We also introduce a 
concept of parallel 
transport, and by using the q-pairing, extended to module valued forms, we 
construct a linear map from the algebra to the endomorphisms of the module. 
This map corresponds to the transition from coordinates to covariant coordinates as  introduced in \cite{Madore:2000en} in a somewhat different context. 
Similar to the discussion there, 
the curvature of this map turns out to be related to the curvature of the 
connection. The map can be used to construct gauge invariant local functionals (''local observables'') of the theory. The exposition here extends an earlier publication~\cite{hesselberg}

Let $ H$ be a right module over $ A$. We set 
\begin{equation}
        \Lambda^{n}( A, H) := H\otimes_{ A}\Lambda^{n}( A)  =  H\otimes A^{\otimes n} \ ,
	        \label{module}
\end{equation}
which is a right $A$-module. For $\Phi \in \Lambda^{n}( A, H)$, $\omega \in \Lambda^m(A)$, we have
$\Phi\otimes_{ A}\omega \in \Lambda^{n+m}( A, H)$, so
$$
\Lambda( A, H):=\bigoplus_{n=0}^{\infty}\Lambda^{n}( A, H)
$$ 
is a right $\Lambda( A)$-module with respect to the action $\Phi\omega:= \Phi\otimes_{ A}\omega$. 
As on 
$\Lambda^{n}(A)$, we have the multiplication maps $m_{k}:\Lambda^{n}( A, H)\to\Lambda^{n-1}( A, H)$ for $k=0,\dots, n-1$, given by linearly extending 
($\Phi \in H$, $a_j \in A$)
\begin{eqnarray*}
m_{0}(\Phi\otimes\bigotimes_{i=1}^n a_{i})
&:=&\Phi a_{1}\otimes\bigotimes_{i=2}^n a_{i}
\\
m_k(\Phi\otimes\bigotimes_{i=1}^n a_{i})
&:=&\Phi \otimes a_{1} \otimes \cdots \otimes a_k a_{k+1}\otimes a_n \ 
\mbox{ for } 1\leq k \leq n-1
\end{eqnarray*}

The q-pairing has a straightforward extension to a 
pairing of $\Lambda^{n}( A, H)$ with $\Lambda^{n}( A)$ with values 
in $ H$, by
\begin{equation}\label{eq:q-pairing-ext}
	\langle \Phi\otimes a_{1}\otimes\cdot\cdot\otimes a_{n},b_{0}\otimes\cdot\cdot\otimes 
	b_{n}\rangle :=\Phi b_{0}\prod_{i=1}^n a_{i}b_{i} \ .
\end{equation}

We now consider the set of $H$-valued (universal) $n$-forms 
\[
\Omega^{n}( A, H) := H\otimes_A \Omega^n(A) \ , 
\]
which is a right $A$-module. For $\Phi \in \Omega^{n}( A, H)$, $\omega \in \Omega^m(A)$, we have $\Phi\otimes_{ A}\omega \in \Omega^{n+m}( A, H)$, so 
$$
\Omega( A, H):=\bigoplus_{n=0}^{\infty}\Omega^{n}( A, H)
$$ 
is an $\Omega(A)$-module with respect to the action $\Phi\omega:= \Phi\otimes_{ A}\omega$. In particular, we can identify $\Omega^1(A,H)\otimes_A \Omega^n(A)$ with $\Omega^{n+1}(A,H)$.

Once again, $\Omega( A, H)$ is the submodule of $\Lambda(A,H)$ 
which is given by the intersection of the kernels of the multiplication 
maps $m_k$.

A {\itshape universal covariant differential (a universal connection)} on $H$ is a linear map
\begin{equation}
        D: H\rightarrow \Omega^1(A,H)
        \label{connection}
\end{equation}
satisfying the Leibniz rule
\begin{equation}
        D(\Phi a)=(D\Phi)a+\Phi  d a \ .
        \label{covariant Leibniz}
\end{equation}
for all $\Phi \in H$, $a \in A$. 

Observe that a right $H$-module must be projective in order to admit for a 
universal connection \cite[Prop 8.3]{GraciaBondia:2001tr}.

If $D$ is a universal connection, it has an extension (still denoted by \(D\))
to $\Omega(A,H)$, which is uniquely fixed by the 
requirement
\[
D(\Phi \alpha) := (D\Phi ) \alpha + \Phi d\alpha
\]
for all $\Phi \in H$, $\alpha \in \Omega(A)$. Recall here that $\Phi \alpha = \Phi \otimes_A \alpha$, and that $(D\Phi ) \alpha = (D\Phi ) \otimes_A  \alpha$, and that we have $\Omega^m(A,H)\otimes_A \Omega^n(A)=\Omega^{m+n}(A,H)$.
For $\Phi \in \Omega^n(A,H)$, and $\alpha \in \Omega(A)$, we then have the graded Leibniz rule
\begin{equation}
D(\Phi\alpha) := (D\Phi) \alpha + (-1)^n  \Phi d\alpha
\end{equation}

We will now show that the map $D^2:\Omega(A,H) \rightarrow \Omega(A,H)$ is a right $\Omega(A)$-module homomorphism with $D^2:\Omega^n(A,H) \rightarrow \Omega^{n+2}(A,H)$.

To that end, let $\Phi\in\Omega^{n}(A,H)$ and $\psi\in\Omega( A)$. Then 
      from the graded Leibniz' rule and $ d^2=0$ we get
      \begin{eqnarray}\nonumber
        D^2(\Phi\psi) & = & D\left((D\Phi)\psi+(-1)^n\Phi 
         d\psi\right)\\\nonumber
         & = & (D^2\Phi)\psi+(-1)^{n+1}(D\Phi) d\psi+
                     (-1)^n(D\Phi) d\psi\\\label{D2OmegaMod}
         & = & (D^2\Phi)\psi.
      \end{eqnarray}

By analogy with differential geometry, $F:=D^2$ is called the curvature of the 
connection $D$.

Making use of the inclusion of $\Omega^n$ in $\Lambda^n$ 
and the $q$-pairing, we now introduce 
two new concepts: the parallel transport associated to covariant derivative,
and a right module map  which generalises the notion of the covariant coordinates introduced in \cite{Madore:2000en}.

Let $S: H \rightarrow \Lambda^1(A,H)$ be a linear map, then we associate a linear map $U: H \rightarrow \Lambda^1(A,H)$ to $S$ by 
\begin{equation}\label{def:Ugeneral}
U(\Phi) := S\Phi  + \Phi\otimes 1
\end{equation}

Given a universal connection $D: H \rightarrow \Omega^1(A,H)$, 
the linear map $U: H \rightarrow \Lambda^1(A,H)$ associated with \(D\) 
is called  the  {\itshape parallel transport} (along $D$).

A related notion of parallel transport has been introduced in the context of lattice gauge theory \cite{mueller-hoissen,pruestel}

\begin{prop}\label{prop:Umodulemap}
The linear map $U$ associated to a linear map $S: H \rightarrow \Lambda^1(A,H)$ 
in the sense of (\ref{def:Ugeneral}) is a right module map if and only $S$ satisfies the Leibniz rule w.r.t. the universal calculus.

\end{prop} 

Observe that we do not assume that $S$ takes values in $\Omega^1(A,H)$.

\begin{proof} 
Let $S$ satisfy the Leibniz rule, then for $\Phi\in H$, and $a\in A$, we have
             \begin{eqnarray*}
                  U(\Phi a)&=&S(\Phi a)+\Phi a\otimes{\bm 1}
		  \ = \ 
		  (S\Phi)a+\Phi  d a + \Phi a\otimes{\bm 1}\\
                               &=&(U\Phi -\Phi\otimes{\bm 1})a 
			       +\Phi({\bm 1}\otimes 
                                 a-a\otimes{\bm 1})+\Phi a \otimes{\bm 1}\\
                                &=&(U\Phi)a     
                         \label{module homomorphism}
               \end{eqnarray*}
On the other hand, let $U$ be a right module map, then 
$
S(\Phi):=U(\Phi)-\Phi \otimes{\bm 1}
$ 
satisfies the Leibniz rule w.r.t. the universal calculus,
\[
S(\Phi a)=U(\Phi) a -\Phi a \otimes{\bm 1} = (S(\Phi) + \Phi \otimes {\bm 1}) a
-\Phi a \otimes{\bm 1} = S(\Phi)a + \Phi d a
\]

\end{proof}

If $D$ is a universal connection, the parallel transport $U$ along $D$ is a right module map (by the Proposition above) and moreover satisfies $m_0 \circ U = 1$, since $D$ takes values in $\Omega^1(A,H)$, so that $m_0 \circ D= 0$ whence $m_0\circ U (\Phi)= m_0(\Phi \otimes 1) = \Phi$. 
Conversely, if $U$ is a right module map that satisfies $m_0 \circ U= 1$, then 
$S :=U-\cdot \otimes  1$ 
is a universal connection.

It follows that the parallel transport splits the exact sequence 
      \begin{equation}
      0 \longrightarrow \Omega^1(A,H) \hookrightarrow H \otimes A \stackrel{m_0}{\longrightarrow} H \longrightarrow 0
      \end{equation}
This splitting map  was used to prove that a module admitting for a universal connection must be projective in \cite[Prop.\ 8.3]{GraciaBondia:2001tr}.

Using the notion of the parallel transport, we now extend $D$ to all of $\Lambda(A,H)$ in the same spirit as in the definition of the exterior differential on $\Lambda(A)$ in the previous section.

\begin{Def}\label{covariant differential on n-forms}
      On $\Lambda^{n}( A, H)$, the covariant differential is defined, as a 
      linear map into $\Lambda^{n+1}( A, H)$, by
      \begin{eqnarray}\nonumber
              D\Phi\otimes\bigotimes_{i=1}^n a_{i}&=&U\Phi\otimes\bigotimes_{i=1}^n a_{i}\\
            &  &+\sum_{k=1}^{n+1}(-1)^{k}\Phi\otimes\bigotimes_{i=1}^{k-1}a_{i}
              \otimes{\bm 1}\otimes\bigotimes_{i=k}^n a_{i} \ .
       \end{eqnarray}
\end{Def}

The map defined above satisfies the graded Leibniz rule,
\begin{equation}\label{eq:cov-Leibn}
              D(\Phi\psi)=(D\Phi)\psi+(-1)^n \Phi  d\psi \ , \ 
              \Phi\in\Lambda^{n}( A, H),\psi\in\Lambda( A)\ .
      \end{equation}
To see this,
      it suffices to prove the proposition for 
      $n=0$. Let $\psi=\bigotimes_{i=0}^m a_{i}$. Then $\Phi \psi=\Phi 
       a_{0}\otimes\bigotimes_{i=1}^m a_{i}$ and
       \begin{equation}
               D(\Phi\psi)=U\Phi a_{0}\otimes\bigotimes_{i=1}^m a_{i}+
               \sum_{k=1}^{m+1}(-1)^k\Phi a_{0}\otimes\bigotimes_{i=1}^{k-1}a_{i}
               \otimes{\bm 1}\otimes\bigotimes_{i=k}^m a_{i}
               \label{}
        \end{equation}
        On the other hand, 
        \begin{eqnarray*}
                (D\Phi)\psi&=&(U\Phi-\Phi\otimes{\bm 1})\psi \\
                &=&U\Phi a_{0}\otimes\bigotimes_{i=1}^m a_{i}- 
                \Phi\otimes\bigotimes_{i=0}^m a_{i}
        \end{eqnarray*}
        hence
        \begin{displaymath}
        	D(\Phi\psi)-(D\Phi)\psi=\Phi d\psi
        \end{displaymath}       
        which proves the assertion. 

By the same argument used for proving \eqref{D2OmegaMod}, we see that
$D^2:\Lambda( A,H)\rightarrow \Lambda( A,H)$ is a right $\Lambda( A)$ module homomorphism.

We now turn to the generalisation of covariant coordinates.

\begin{prop}
For any $a \in A$, 
\begin{equation}
L(a)\Phi:=\Phi a -\langle D\Phi, d a\rangle = 
\Phi a -\langle U\Phi, d a\rangle 
\end{equation}
is a right module map $L(a):H\rightarrow H$.
\end{prop}
\begin{proof} By definition, we have $L(a)\Phi \in H$. To prove that $L(a)$ is a right module homomorphism, we first note that for $b \in A$,
$$
L(a)(\Phi b)=\Phi ba-\langle (D\Phi) b + \Phi  d b, d a\rangle \ .
$$
Now, according to the rules for the q-pairing, we get
$$
\langle (D\Phi) b, d a\rangle =\langle D\Phi, d a\rangle b
$$
and
$$
\langle \Phi  d b, d a\rangle =\Phi [b,a] \ .
$$
Hence $L(a)(\Phi b)=(L(a)\Phi)b$. 
\end{proof}

We now show that in the special case of the quantum spacetime algebra $\mathcal E$, the $L(q^\mu)$'s are the covariant coordinates of~\cite{Madore:2000en}. For a related discussion, see~\cite{Zahn:2006mg}.

Consider the algebra as a module over itself. Pick a covariant derivative, $Da=da+Aa$ where $A$ is a 1-form. Then for the quantum coordinates $q^\mu$, we find
\begin{eqnarray}\nonumber
L(q^\mu)(a) &=& aq^\mu - \langle Da , dq^\mu \rangle 
\ = \  aq^\mu - \langle da , dq^\mu \rangle + \langle Aa , dq^\mu \rangle 
\\
&=& aq^\mu - [a , q^\mu ] + \langle A , dq^\mu \rangle a
\ = \  q^\mu a + \langle A , dq^\mu \rangle a
\end{eqnarray}
where we have used the rules \eqref{eq:add-properties} for the $q$-pairing to pass from the first to the second line. Now, if $A= A^{(1)}_\nu dq^\nu A^{(2)}_\nu$ (in Sweedler's notation), we indeed find, 
\begin{equation}
L(q^\mu)(a) 
= (q^\mu  + i Q^{\mu\nu} A^{(1)}_\nu A^{(2)}_\nu ) \, a 
\end{equation}
by application of the rules \eqref{eq:add-properties}.

It should be noted that $L$ is not multiplicative. Instead we have:
\begin{prop}
Let $a,b, \in A$. Then the algebraic curvature 
\begin{equation}
R_{L}(a,b):=L(a)L(b)-L(ab)
\end{equation}
of \(L\) is related to the geometric curvature 
$F=D^2$ by
\begin{equation}
  	R_{L}(a,b)\Phi=\langle F\Phi, d a d b\rangle 
\end{equation}
for all $\Phi \in H$.
\end{prop}
\begin{proof} 
Let $D\Phi=\sum\Phi_{i} d c_{i}$. Then
\begin{displaymath}
	L(b)\Phi=\Phi b -\sum\Phi_{i}\langle d c_{i}, d b\rangle
\end{displaymath}
and
\begin{displaymath}
	L(a)L(b)\Phi=(\Phi a-\langle D\Phi, d a\rangle )b 
	-\sum(\Phi_{i}a-\langle D\Phi_{i}, d a\rangle)\langle d c_{i}, d b\rangle \ .
\end{displaymath}
We have $\sum\Phi_{i}a\langle d c_{i}, d b\rangle=\sum\Phi_{i}\langle d 
c_{i},a d b\rangle = \langle D\Phi,a d b\rangle$ and $\langle D\Phi_{i}, d a\rangle
\langle d c_{i}, d b\rangle=\langle D^2\Phi, d a d b\rangle$, hence
\begin{displaymath}
	L(a)L(b)\Phi=L(ab)\Phi + \langle F\Phi, d a d b\rangle \ .
\end{displaymath}
\end{proof}
Gauge invariant quantities are now obtained in terms of a trace 
on the algebra of endomorphisms of $ H$. For instance, the evaluation 
of the $a,b$-component of the field 
strength smeared with a ``test function'' represented by some suitable 
element $c$ of $ A$ may be defined as
\begin{equation}
\Tr L(c)R_{L}(a,b).
\end{equation}
Therefore, we can interpret the above expressions
as providing noncommutative analogues of {\itshape local} gauge invariant quantities.

\section{The spectrum of geometric operators on Quantum Spacetime}

We will now use the concepts developed in the previous section to define volume, area and distance operators associated with the 
coordinates $q^\mu$ uniquely associated with the quantum spacetime algebra 
\(\mathcal E\), as described in the introduction. 
To define their spectra we regard 
the homogeneous elements of the universal differential 
algebra as elements of tensor powers (of \(\mathcal Z\)-modules) 
of the quantum spacetime.

We will make use of the noncommutative analogue of the wedge product: For any two tensors $A=(A^{ \mu_1 \dots \mu_n})$,  $B=(B^{ \nu_1 \dots \nu_m})$ with entries from a noncommutative algebra, we define the tensor
\begin{equation}
A \wedge B = (A^{[\mu_1 \dots \mu_n } B^{ \nu_1 \dots \nu_m]})
\end{equation}
where the brackets denote total antisymmetrisation as usual. 
In the highest rank case, we will consider, unless otherwise stated, the 
component 
$
(A \wedge B)^{0123}
$
and denote it (also) by $A\wedge B$.

\subsection{The four-volume operator}

We start by defining the 4-volume operator as an element of $\Lambda^4(\mathcal E)$, 
\begin{equation} \label{somelabel}
V =  dq \wedge dq \wedge dq \wedge dq = 
\epsilon_{\mu\nu\rho\sigma}dq^{\mu}dq^{\nu}dq^{\rho}dq^{\sigma} \ ,
\end{equation} 
so that the operator $V$ is an element of the 5th tensor power. 
It is crucial to observe that, while the product operation 
used in \eqref{somelabel} is defined
in  \(\Lambda(\mathcal E)\), the spectrum of \(V\) has to be computed in
the completed tensor product \(\mathcal E^{\otimes 5}\), 
equipped with its own C*-algebra structure.  
We interpret \(V\) as a function of 5 mutually commuting independent quantum 
coordinates $q^{\mu}_j=1^{\otimes j-1}\otimes q^{\mu}\otimes 1^{\otimes 5-j-1}$, $j=1,\ldots,5$, denoting the vertices of a 4 dimensional simplex. 

We now represent 
$V$ as an operator on the 5th tensor power of the 
representation space of the standard representation and analyse its spectrum.
The explicit form of $V$ is
\begin{equation}
V=\epsilon_{\mu\nu\rho\sigma}(q^{\mu}_2-q^{\mu}_1)(q^{\mu}_3-q^{\mu}_2)(q^{\mu}_4-q^{\mu}_3)(q^{\sigma}_5-q^{\sigma}_4) \ .
\end{equation}
We expand this product into a sum of products of coordinates,
\begin{equation}
V=\sum \pm \epsilon_{\mu\nu\rho\sigma}q_{j_1}^{\mu}q_{j_2}^{\nu}q_{j_3}^{\rho}q_{j_4}^{\sigma}
\end{equation}
where the sum runs over all choices of $j_1,\ldots j_4\in\{1,\ldots,5\}$ with $j_k\in\{k,k+1\}$,
and where the overall sign is $(-1)^{\sum_k j_k}$.
We may decompose the sum into three parts corresponding to the number of coincidences in the indices  $j_k$. 
The only term with 2 coincidences corresponds to
\begin{equation}
(j_1,j_2,j_3,j_4)=(2,2,4,4)
\end{equation}
This term is
\begin{equation}
\epsilon_{\mu\nu\rho\sigma} q_2^{\mu}q_2^{\nu}q_4^{\rho}q_4^{\sigma} =\frac14\epsilon_{\mu\nu\rho\sigma}[q_2^{\mu},q_2^{\nu}][q_4^{\rho},q_4^{\sigma}] \ .
\end{equation}
We performed the tensor product over the centre, which means that the commutators with different lower indices can be identified, and find
\begin{equation}
\label{QwQ}
=-\frac14 Q\wedge Q=-2\eta
\end{equation}
where $\eta$ is the central element of the algebra \(M(\mathcal E)\) of Quantum Spacetime describing the orientation (with spectrum $\pm1$).
 
 The part $A$ with no coincidences corresponds to the 5 subsets of $\{1,\ldots,5\}$ with 4 elements,
 \begin{equation}
A=\sum_{k=1}^5 (-1)^k A_k \ ,\ A_k=\bigwedge_{j\ne k}q_j \ .
\end{equation}
The part with 1 coincidence corresponds to the 10 subsets of $\{1,\ldots,5\}$ with 2 elements
\begin{equation}
B=\sum_{i<k}(-1)^{i+k}B_{ik} \ ,\ B_{ik}=i\frac12 Q\wedge q_i\wedge q_k \ .
\end{equation}

We obtain the decomposition
\begin{equation}\label{eq:decomposition}
V=A+B-2\eta \ .
\end{equation}
Here $A$ and $\eta $ are selfadjoint and $B$ skewadjoint. We now prove, that $A$ and $B$ commute, i.e. $V$ is normal.

We have
\begin{equation}
[q_k^{\mu},B_{ij}]=Q^{\mu\cdot}\wedge Q\wedge(\delta_{ik}q_j-\delta_{jk}q_i) \ .\end{equation}
For any 4-vector $a$ it holds
\begin{equation}
Q^{\mu\cdot}\wedge Q\wedge a=\epsilon_{\nu\lambda\rho\sigma}Q^{\mu\nu}Q^{\lambda\rho}a^{\sigma}
=2(Q^{\mu\nu}(\ast Q)_{\nu\sigma})a^{\sigma}
\end{equation}
where \(*Q\) denotes the Hodge dual of the commutator tensor 
\eqref{eq:hodge_of_Q}.

But by the relations defining quantum spacetime we have
\begin{equation}\label{eq:hodge_Q}
Q^{\mu\nu}(\ast Q)_{\nu\sigma}=\frac14 Q^{\nu\rho}(\ast Q)_{\nu\rho}\delta^{\mu}{}_{\sigma}=\eta \delta^{\mu}{}_{\sigma}
\end{equation}
hence we arrive at
\begin{equation}
[q_k^{\mu},B_{ij}]=\eta(\delta_{ik}q_j^{\mu}-\delta_{jk}q_i^{\mu}) \ .
\end{equation}
Thus $\mathrm{ad}B_{ij}$ acts on the coordinates $q_k$ as $\eta$ times the $ij$-generator of the Lie algebra of $\mathrm{SO}(5)$.
It follows that $\mathrm{ad}B$ generates a 1-parameter subgroup of $\mathrm{SO}(5)$.
Now 
\begin{equation}
\mathrm{ad}B(\sum_{k=1}^5 q_k)=\eta \sum_{i<j}\sum_k (-1)^{i+j}(\delta_{ik}q_j-\delta_{jk}q_i)=\eta\sum_{i<j}(-1)^{i+j}(q_j-q_i)=0 \ ,
\end{equation}
hence this 1-parameter subgroup is contained in the stabiliser $N$ of the vector $(1,1,1,1,1)^{T}$ in the fundamental representation of $\mathrm{SO}(5)$.

Now $A$ can be written in the form \footnote{Note that in our matrices
  here and in the following only entries on the same row may fail to
commute with each other; as a consequence, we can develop
indifferently by rows or by columns.}
\begin{equation}\label{eq:det1}
A=\sum_j(-1)^j\bigwedge_{i\ne j}q_i = \mathrm{det}
\left(
\begin{array}{cccccc}
 1 & q_1^0  & q_1^1&q_1^2& q_1^3 \\
 1 & q_2^0  & q_2^1&q_2^2& q_2^3 \\
 1 & q_3^0  & q_3^1&q_3^2& q_3^3 \\
 1 & q_4^0  & q_4^1&q_4^2& q_4^3 \\
 1 & q_5^0  & q_5^1&q_5^2& q_5^3
\end{array}
\right)
\end{equation}
If we multiply the matrix on the right hand side from the left by a $\mathrm{SO}(5)$-matrix $R$,
the left hand side does not change. If $R\in N$, the right hand side of (\ref{eq:det1}) assumes the form
\begin{equation}
\text{det}
\left(
\begin{array}{ccccc}
 1 & {q'}_1^0& {q'}_1^1& {q'}_1^2& {q'}_1^3  \\
 1 & {q'}_2^0& {q'}_2^1& {q'}_2^2& {q'}_2^3  \\
 1 & {q'}_3^0& {q'}_3^1& {q'}_3^2& {q'}_3^3  \\
 1 & {q'}_4^0& {q'}_4^1& {q'}_4^2& {q'}_4^3  \\
 1 & {q'}_5^0& {q'}_5^1& {q'}_5^2& {q'}_5^3  \\
 
\end{array}
\right)
\end{equation}
where ${q'}_j=R_j{}^kq_k$. Hence $A$ commutes with all linear combinations of $B_{ij}$ corresponding to elements of the Lie algebra of $N$. In particular, $A$ commutes with $B$.

Since $V$ is invariant under proper Lorentz transformations and changes sign under parity, we 
may evaluate it on a point $\sigma\in\Sigma$ of the spectrum of $Q$. We choose the representation \eqref{basic_rep} where
$\sigma=(\vec{e},\vec{m})$ with $\vec{e}=\vec{m}=(1,0,0)$. 
We represent the operators 
$q_j^{\mu}$ on the Hilbert space 
$H^{\otimes 5}\otimes H^{\otimes 5}$ by
\begin{subequations}
\begin{align} 
q_k^0 &= Q_k  \otimes I,\\ 
q_k^1 &= P_k  \otimes I,\\
q_k^2 &= I        \otimes Q_k,\\
q_k^3 &=I        \otimes P_k,
\end{align} 
\end{subequations}
where the index $k$ refers to the tensor factor in $H^{\otimes 5}$.

In this representation, $A$ assumes the form
\begin{align}\nonumber
A=&\mathrm{det}
\left(
\begin{array}{ccccc}
 1 & Q_1\otimes I & P_1 \otimes I& I \otimes Q_1 & I \otimes P_1  \\
 1 & Q_2\otimes I & P_2 \otimes I& I \otimes Q_2 & I \otimes P_2  \\
 1 & Q_3\otimes I & P_3 \otimes I& I \otimes Q_3 & I \otimes P_3  \\
 1 & Q_4\otimes I & P_4 \otimes I& I \otimes Q_4 & I \otimes P_4  \\
 1 & Q_5\otimes I & P_5 \otimes I& I \otimes Q_5 & I \otimes P_5
\end{array}
\right) =\\
\label{eq:det2}
=&\frac14 \sum_{i,j,k,l,m=1}^5 \epsilon_{ijklm}M_{jk}\otimes M_{lm}
\end{align}
where $M_{jk}=Q_{j}P_{k}-Q_{k}P_{j}$ is the generator of rotations in in the $(jk)$-plane of  $\mathbb{R}^5$. We see this by developing first the determinant along the first column,
in terms of the determinants of the 4 x 4 complementary minors, and then
developing each of those in terms of the determinants of their 2 x 2
minors from the two first columns time the determinants of their 2 x 2
complementary minors.

We may now multiply the matrix in (\ref{eq:det2}) from the left with an $\mathrm{SO}(5)$ matrix, such that the column vector $(1,1,1,1,1)^T$ is mapped onto the vector 
$(0,0,0,0,\sqrt{5})^T$. This will not change the operator $A$.  The evaluation of the determinant gives
\begin{equation}
A=\frac{\sqrt{5}}{4} \sum_{j,k,l,m=1}^4 \epsilon_{jklm}M'_{jk}\otimes M'_{lm} \ ,
\end{equation}
where $M'_{jk}$ is the generator of rotations in the $(jk)$-plane in the transformed coordinates, $j,k=1,\ldots,4$.

Let now $\vec{B}=(M'_{23},M'_{31},M'_{12})$, $\vec{D}=(M'_{14},M'_{24},M'_{34})$.
Then $\vec{L}^{\pm}=\frac12(\vec{B}\pm\vec{D})$ are mutually commuting generators of $\mathrm{SU}(2)$, in terms of which $A$ can be written as
\begin{equation}
A=2\sqrt{5}(\vec{L}^+\otimes \vec{L}^+ -\vec{L}^-\otimes \vec{L}^-)
\end{equation}
The spectrum of $A$ can now be obtained from an analysis of the representations of $\mathrm{SO}(4)$ in terms of the representations of the two commuting $\mathrm{SU}(2)$'s. They are labeled by pairs of spins $(j_+,j_-)$ where $j_++j_-$ must be integer. We thus obtain
\begin{equation}
\mathrm{spec}(A)\subset \sqrt{5}\mathbb{Z} \ .
\end{equation}

In view of (\ref{eq:decomposition}) and the fact that $B$ is skewadjoint and commutes with $A$ we find that the spectrum of $V$ has at least a distance $\sqrt{5}-2$ from the origin.

For a full determination of the spectrum of $V$ we now determine the spectrum of $B$. $\eta B$ is a representative of the Lie algebra element
\begin{equation}
b=
\left(
\begin{array}{rrrrr}
 0 & -1  & 1  & -1  & 1 \\
 1 &  0  & -1 &  1  &-1 \\
 -1&  1  &  0 & -1  & 1 \\
 1 &  -1 & 1  &  0  & -1\\
 -1&  1  & -1 &  1 & 1 \\
\end{array}
\right)
\end{equation}
of the Lie algebra
$\mathfrak{so}(5)$ of $\mathrm{SO}(5)$ under the natural representation  on 
$H^{\otimes 5}\otimes H^{\otimes 5}$.
The characteristic polynomial of $b$ is
\begin{equation}
\mathrm{det}(b-\lambda 1)=(-\lambda)^5+10(-\lambda)^3+5(-\lambda)
\end{equation}
with the roots
\begin{equation}\label{eigen}
\lambda_1=0\ ,\quad\lambda_{2,3}=\pm i \sqrt{5-{2\sqrt{5}}},\quad \lambda_{4,5}=\pm i\sqrt{5+{2\sqrt{5}}} \ . 
\end{equation}
We conclude that  $B$ has pure point spectrum with eigenvalues of the form 
$n\lambda_2+m\lambda_4$ with integers $n$ and $m$. Since $\lambda_4/\lambda_2=\sqrt{5}+2$ is irrational, the eigenvalues are dense in $i\mathbb{R}$, hence $\mathrm{spec}(B)=i\mathbb{R}$.

The point spectrum of the volume operator $V$ is thus contained in the set
\begin{equation}\label{spectrum}
S=\pm 2 +\mathbb{Z}\sqrt{5}+
i\left(\mathbb{Z}\sqrt{5-{2\sqrt{5}}} +\mathbb{Z}\sqrt{5+{2\sqrt{5}}}\;\right) \ .
\end{equation}
It is easy to see, that actually all points of $S$ belong to the 
point spectrum 
of $V$.  Namely, the eigenvalue $a$ of $A$ is determined by fixing the eigenvalues for $|\vec{L}^{\pm}\otimes I|^2$,  $|I\otimes \vec{L}^{\pm}|^2$ and $|\vec{L}^{\pm}\otimes I+I\otimes\vec{L}^{\pm}|^2$. If we call the corresponding spin quantum numbers $j_1^{\pm}$, $j_2^{\pm}$ and $j_1^{\pm}+j_2^{\pm}-m^{\pm}$ with $m^{\pm}=0,1,\ldots,2\,\mathrm{min}(j_1^{\pm},j_2^{\pm})$, we obtain
\begin{align}\nonumber
a=&\sqrt{5}\big((j_1^{+}+j_2^{+}-m^{+})(j_1^{+}+j_2^{+}-m^{+}+1)-j_1^{+}(j_1^{+}+1)-j_2^{+}(j_2^{+}+1)+\\
\nonumber
&-(j_1^{-}+j_2^{-}-m^{-})(j_1^{-}+j_2^{-}-m^{-}+1)+j_1^{-}(j_1^{-}+1)+j_2^{-}(j_2^{-}+1)\big)\\
\nonumber
=&\sqrt{5} \big(2j_1^{+}j_2^{+}-2m^{+}(j_1^{+}+j_2^{+})+m^{+}(m^{+}-1)-2j_1^{-}j_2^{-}+\\
&+2m^{-}(j_1^{-}+j_2^{-})-m^{-}(m^{-}-1)\big)
\end{align}
We now choose $j_2^{\pm}=\frac12$. Then $m^{\pm}=0,1$.
We set $m^{\pm}=0$ and obtain $a=\sqrt{5}(j_1^{+}-j_1^{-})$, thus all integer multiples of $\sqrt{5}$ occur as eigenvalues of $A$.

In order to obtain the eigenvalues of $B$ in a given eigenspace of $A$ we determine the decomposition of the matrix $b$ into the linear combination $b=\alpha_{+}b^{+}+\alpha_{-}b^{-}$ of two standard generators $b^{+},b^{-}$ of the two commuting $\mathfrak{su}(2)$ sub-Lie-algebras in the fundamental representation of $\mathfrak{so}(4)$ in the 4 dimensional subspace of $\mathbb{R}^5$ orthogonal to $(1,1,1,1,1)$. Up to a unitary transformation, these generators are given by
\begin{equation}
b^{\pm}=\frac12(M'_{23}\pm M'_{14})
\end{equation}
and they satisfy the relations
\begin{equation}
(b^{\pm})^2=-\frac14 , \ , b^{+}b^{-}=b^{-}b^{+}=:\Gamma \ , \Gamma^2=1 \ .
\end{equation}
The coefficients $\alpha_{\pm}$ can be determined from the characteristic equation
\begin{equation}
b^4+10b^2+5=0 \ .
\end{equation}
Inserting the decomposition of $b$ and using the fact that the even powers of $b$ are linear combinations of 1 and $\Gamma$, we find 
\begin{equation}
\alpha_+^2+\alpha_-^2=20
\end{equation}
and
\begin{equation}
\frac{(\alpha_+^2+\alpha_-^2)^2}{16}+ \frac{\alpha_+^2\alpha_-^2}{4}-\frac{5}{2}(\alpha_+^{^2}+\alpha_-^{2})+5=0 
\end{equation}
with the solution 
\begin{equation}
\alpha_{\pm}^2=10\pm{2\sqrt{5}} \ .
\end{equation}
In a representation of $\mathfrak{so}(4)$ with spin quantum numbers $(j^{+},j^{-})$ with $j^{+}+j^{-}\in\mathbb{Z}$ the representative of $b$ thus assumes the eigenvalues
\begin{equation}
i(k^{+}\alpha_{+}+k^{-}\alpha_-)
\end{equation}
with $k^{\pm}=j^{\pm}-l^{\pm}$, $l^{\pm}=0,1,\ldots, 2j^{\pm}$.
In terms of the eigenvalues $\lambda_2$ and $\lambda_4$ of $b$ (see (\ref{eigen}),
\begin{equation}
\lambda_{2,4}=i\sqrt{5\pm{2\sqrt{5}}}
\end{equation}
using
\begin{equation}
\frac{i}{2}(\alpha_+\pm\alpha_-)=\sqrt{5\pm{2\sqrt{5}}}
\end{equation}
we find
\begin{equation}
i(k^{+}\alpha_{+}+k^{-}\alpha_-)=i(k^{+}+k^{-})\lambda_4+i(k^{+}-k^{-})\lambda_2\ .
\end{equation}
We see that every linear combination of $\lambda_2$ and $\lambda_4$ with integer coefficients can be obtained by an appropriate choice of $j^{\pm}$ and $l^{\pm}$. 
We conclude that the set $S$ in (\ref{spectrum}) is the point spectrum of 
$V$.
 
The multiplicity of this point spectrum is uniformly infinite. For our 
operator commutes with the group of joint translations of the five vertices, 
and the relevant  representation of this  group 
has no non zero finite dimensional subrepresentations.

\subsection{The three-volume operators}

The four components of the {\itshape 3-volume operator} can be expressed, in 
the formalism used above, as
\begin{equation}
V_\sigma=\epsilon_{\mu\nu\rho\sigma}dq^\mu dq^\nu dq^\rho =A_\sigma+iB_\sigma
\end{equation}
where
\begin{subequations}
\begin{align}
A_\sigma&=\frac 16
\mathrm{det}
\left(
\begin{array}{cccc}
1&q_1^\mu&q_1^\nu&q_1^\rho\\ 
1&q_2^\mu&q_2^\nu&q_2^\rho\\
1&q_3^\mu&q_3^\nu&q_3^\rho\\
1&q_4^\mu&q_4^\nu&q_4^\rho\\
\end{array} \right)\epsilon_{\mu\nu\rho\sigma},\\
\nonumber
B_\sigma&=\frac 12Q^{\mu\nu}(q_1^\rho-q_2^\rho+q_3^\rho-q_4^\rho)\epsilon_{\mu\nu\rho\sigma}=\\
&={\tilde q}_{1\sigma}-{\tilde q}_{2\sigma}+{\tilde q}_{3\sigma}-{\tilde q}_{4\sigma}.\label{B_as_q_prime}
\end{align}
\end{subequations}
The operator \(V_\sigma\) is {\itshape normal}. For, by the Leibniz rule, the 
commutator of the real with 
the imaginary part can be computed as $1/2$
the sum of the determinants of the 
$4$ by $4$ matrices appearing in the last equation, where the $2$nd, 
or the $3$rd, or the $4$th column are replaced by the commutators of 
their entries with $B_\sigma$. Thus the first term, for 
instance, will give
\begin{equation}
\mathrm{det}
\left(
\begin{array}{cccc}
 1 & [q_1^\mu, B_\sigma]&q^\nu_1&q^\rho_1  \\
 1 & [q_2^\mu, B_\sigma]&q^\nu_2&q^\rho_2  \\
 1 & [q_3^\mu, B_\sigma]&q^\nu_3&q^\rho_3  \\
 1 & [q_4^\mu, B_\sigma]&q^\nu_4&q^\rho_4  \\
\end{array} \right)\epsilon_{\mu\nu\rho\sigma}.
\end{equation}
For each \(\sigma\) the term \([q_j^\mu,B_\sigma]\epsilon_{\mu\nu\rho\sigma}\) 
takes the form 
\[
(-1)^j[q_j^\mu,{\tilde q}_{j\sigma}]\epsilon_{\mu\nu\rho\sigma}=
(-1)^j{\delta^\mu}_\sigma\epsilon_{\mu\nu\rho\sigma}=0.
\]
Analogous computations can be performed for the other columns, so that
our operator is normal.

It follows that, as operators affiliated to $\mathcal E ^{\otimes _{\mathcal 
Z} 4}$, our $V_\sigma$ fulfill \
\begin{equation}
\sum_\sigma V_\sigma^* V_\sigma=\sum_\sigma(A_\sigma^2+B_\sigma^2)\geqslant
\sum_\sigma B_\sigma^2.
\end{equation}
Now, the \({\tilde q}\)s obey the relations
\begin{equation}\label{eq:qprime_rels}
[{\tilde q}_\sigma,{\tilde q}_\rho]=iQ^{-1}_{\sigma\rho},
\end{equation}
so that we see by equation \eqref{B_as_q_prime} that the operators \(B_\sigma\)
obey the commutation relations
\begin{equation}
[B_\sigma,B_\rho]=i4Q^{-1}_{\sigma\rho}.
\end{equation}
Now, \eqref{eq:hodge_Q} shows that \(Q^{-1}\) is related to \(Q\) by the exchange of 
\(\vec e\) and \(\vec m\), up to a change of sign;
therefore the arguments of \cite{Doplicher:1994tu} can be repeated to give 
\[
\sum_\sigma B_\sigma^2\geqslant 8.
\]
We conclude that the square Euclidean length of the three-volume operator
vector is bounded below by \(8\) in Planck units.

It is worth noting that, however, the spectrum of the time component, namely 
of the space 
3-volume
operator, is, as we will show now, the whole complex plane.

That operator can be expressed, in our formalism, as
\begin{align}\nonumber
d \vec q \wedge d \vec q \wedge d \vec q & \equiv \mathrm{det}
\left(
\begin{array}{cccc}
 1 & q_1^1&q_1^2& q_1^3  \\
 1 & q_2^1&q_2^2& q_2^3  \\
 1 & q_3^1&q_3^2& q_3^3  \\
 1 & q_4^1&q_4^2& q_4^3  \\
\end{array}
\right) + (i/2) \sum_{j=1}^{3} m_j \sum_{k=1}^{4} (-1)^{k+1} q_{k}^{j}=\\
&= A_0 + i B_0.
\end{align}
This is a {\itshape normal operator}, whose spectrum will 
contain the spectrum of the 
image in any 
representation; it will then suffice to show that
its spectrum is the whole complex 
plane in the 
irreducible representation \eqref{basic_rep}, where
\begin{equation}
q^1 = P \otimes I, q^2 = I \otimes Q, q^3 = I \otimes P,
\vec  m = (1,0,0);
\end{equation}
so that, in that representation, we have
\begin{equation}
d \vec q \wedge d \vec q \wedge d \vec q = \mathrm{det}
\left(
\begin{array}{cccc}
 1 & P_1 \otimes I& I \otimes Q_1 & I \otimes P_1  \\
  \cdot&\cdot &\cdot &\cdot \\
  \cdot&\cdot &\cdot &\cdot \\
 1 & P_4 \otimes I& I \otimes Q_4 & I \otimes P_4
\end{array}
\right) + (i/2) \sum_{k=1}^{4} (-1)^{k+1} P_{k} \otimes I.
\end{equation}
Furthermore, if we 
develop the determinant in terms of the determinants of the $2 \times 2$ 
minors from the two first columns times the determinants of the 
complementary minors, we find 
\begin{equation}
d \vec q \wedge d \vec q \wedge d \vec q = 
(1/4)\sum \epsilon _{jklm}(P_k - P_j ) \otimes M_{l,m} + (i/2) \sum_{k=1}^{4} 
(-1)^{k+1} P_{k} \otimes I.
\end{equation}
This operator is affiliated to the tensor product of the commutative 
C*-Algebra generated by $P_1 , \dots, P_4 $ with all bounded operators. We can 
evaluate this operator at a point $p = (p_1 , \dots, p_4 )$ of the joint 
spectrum of $P_1 , \dots, P_4 $; re-expressing, as in the previous 
subsection, the $M_{l,m}$ in terms of the independent generators of 
two copies of $SU(2)$, we obtain
\begin{equation}\label{eq:is_that}
\mu  I \otimes  L_{\xi _{+} } ^{+} + \lambda I \otimes 
L_{\xi _{-} } ^{-} + i\eta I,
\end{equation}
where $\xi _{+}, \xi _{-} $ are two unit vectors in three space, and $\mu, 
\lambda, \eta$ are real numbers, all functions of $p$, such that
\begin{gather*}
\mu     \xi_+      = (1/2)(p_3 - p_2 + p_4 - p_1, p_3 - p_2 + p_4 - p_1, p_4 - p_3 
+  p_2 - p_1)  \equiv (u, u, w);\\ \lambda  \xi_- = (1/2)(p_4 
  - p_3 + p_2 - p_1, p_1 - p_3 + p_4 - p_2, p_1 - p_3 + p_4 - p_2) \equiv 
  (w, v, v); \\
\eta (p) = (1/2) \sum_{k=1}^{4} (-1)^{k+1} p_{k} = - 
  (1/2)(p_4 - p_3 + p_2 - p_1) = - w; 
\end{gather*}
 where, as $p = (p_1, \dots, p_4)$ vary in $\mathbb{R} ^{4}$, $(u, v, w)$ span $\mathbb{R} ^{3}$. Since the 
  joint spectrum of $L_{\xi _{+} } ^{+}, L_{\xi _{-} } ^{-}$ includes all 
  pairs of relative integers, we have that the spectrum of our operators 
  includes the set 
\[  
(2u^2 + w^2)^{1/2} \mathbb{Z} + (2v^2 + 
  w^2)^{1/2} \mathbb{Z} - i w, 
\] 
for all possible choices of the real 
  numbers $u, v, w$.  From these explicit expressions, one sees that the 
  last operator has a spectrum which, as $p$ varies in $\mathbb{R} ^{4}$, spans 
  the whole complex plane, and is pure Lebesgue by \eqref{eq:is_that}.

\subsection{The area operators}
The {\itshape area operators} can be discussed separately as the {\itshape 
space-space area operator} and the {\itshape space-time area operator}, 
respectively given by the square roots of the sum of the square moduli of
 
\begin{gather*}
dq^{j} \wedge dq^{k};\\
dq^{j} \wedge dq^{0};
\end{gather*}
now, 
\begin{gather*}
\epsilon_{ljk}dq^{j} dq^{k} = \\
\epsilon_{ljk} (I \otimes q^{j} - 
q^{j} \otimes I) (I \otimes q^{k} -
q^{k} \otimes I) = \\
\epsilon_{ljk} (I \otimes q^j \otimes q^k -  q^j \otimes I \otimes q^k 
+  q^j \otimes q^k \otimes I) + i m_l I \otimes I \otimes I
\end{gather*}
which is a normal operator since $ \vec m $ is central; the sum over 
$l = 1,2,3$ of the square moduli is then bounded below by $(\vec m)^2$, 
which is bounded below by $I$ due to the Quantum Conditions.

Quite similarly,
\begin{align*}
dq^{j} \wedge dq^{0} = &(I \otimes q^j  -  q^j \otimes I)(I \otimes q^0 -  q^0 \otimes I) +\\
 &- (I \otimes q^0  -  q^0 \otimes I)(I \otimes q^j -  q^j \otimes I) =  \\
= &(I \otimes q^j \otimes q^0 -  q^j \otimes I \otimes q^0 +  q^j \otimes q^0 \otimes I)+\\
& - (I \otimes q^0 \otimes q^j -  q^0\otimes I \otimes q^j +  q^0 \otimes q^j \otimes I) + 
i e_j  I \otimes I\otimes I
\end{align*}
are normal operators whose sum of the square moduli is bounded below by $ 
(\vec e)^2$, which is bounded below by $I$ due to the Quantum 
Conditions.

Note that here too normalcy of our operators depends in an essential
way upon the fact that the commutators
of the coordinates are central.

Thus both the space-space and the space-time area operators are bounded 
below by $1$ in Planck units.

\centerline{\Large ***}
\bigskip

\noindent Finally we would like to point out that 
we might also have calculated the ``dual'' length, area, 3- and 4-volume operators, with $q$ replaced by $\tilde q=Q^{-1}q$. These yield the contributions to the metric form discussed in section 2. 
Note that in the case of the highest rank, we find that both $V$ and its dual yield the same contribution to the form, since 
\begin{equation}
d{\tilde q} \wedge d{\tilde q} \wedge d{\tilde q} \wedge d{\tilde q} = 
\eta dq \wedge dq \wedge dq 
\wedge dq,
\end{equation}
and $\eta$ is the element of $\mathcal Z$ equal to plus or minus one on 
$\Sigma_{+}$ or $\Sigma_{-}$ respectively. Hence $\eta$ disappears in the 
square.

\section{Conclusions}

We applied concepts of the universal differential calculus to define several geometric entities on the model of quantum spacetime introduced in \cite{Doplicher:1994tu}. We showed that they can be interpreted in terms of independent events, the underlying classical picture being a characterisation of simplices by their vertices. This makes it possible to represent these entities by operators on a Hilbert space, and we analysed their spectra.

We found that these operators are normal, and that their spectra have an interesting structure, which matches with the general expectations described in the Introduction.
Let us summarise our findings (where, as above, we adopt Planck units): 
\begin{enumerate}
\item We discussed in the Introduction the {\itshape Euclidean and 
Minkowskian Distance Operators} between two independent events. We 
proved that, while the second has pure Lebesgue spectrum, the first is 
bounded below by a constant of order $1$, {\itshape despite the fact that the model is fully Lorentz 
invariant}.
\item The area operator $dq^{\mu}\wedge dq^\nu$ can be split into a spatial and a spacetime part. For both the sum of the absolute squares of the components is bounded below by the unit operator.
\item The spectrum of the space component of the 3 volume is the full complex plane; the sum of the absolute values of the 4 components of the 3 volume operator is, however, bounded below by \(8\).
\item The 4 volume operator has pure point spectrum
\begin{equation}\label{eq:spec_V}
\text{spec}_{pp}(V)=
S=\pm 2 +\mathbb{Z}a_+a_-+
i\left(\mathbb{Z}a_+ +\mathbb{Z}a_-\right).
\end{equation}
where
\[
a_\pm=\sqrt{5\pm 2\sqrt 5},
\]
so that, of course,
\begin{equation}
\text{spec}(V)=\overline{\text{spec}_{pp}(V)}=
\pm 2 +\mathbb{Z}\sqrt 5+i\mathbb R.
\end{equation}
Figure 1 shows the points in \(\text{spec}_{pp}(V)\) with small real part, 
when the integer coefficients of \(a_+,a_-\) have sum of the moduli not 
exceeding \(3\). The left and right sides of each pair of close columns
refers to the \(-\) and \(+\) sign in equation \ref{eq:spec_V},
and larger dots refer to smaller absolute values of the coefficients. 

The translates of \(S\) by real integers form a ring, so that
finite tensor powers and direct sums of copies of \(V\) fulfil
\[
\text{spec}_{pp}(V\otimes\dotsm\otimes V),\;\;
\text{spec}_{pp}(V\oplus\dotsm\oplus V)\subset\text{spec}_{pp}(V)+\mathbb Z.
\]
The combinations of \(1,a_-,a_+,a_-a_+\) with rational coefficients form
a field. Presently we have no interpretation of these facts.

\begin{center}\input{code}

{\footnotesize
Fig.1 Generation of the pure point spectrum of the 4-volume operator.
}

\end{center}
\end{enumerate}

We have seen that the appropriate mathematical frame to establish these results
is a combination of the universal differential calculus on one side and the
relevant C*-algebraic structures on the other.

This mathematical frame allowed us to recognise the Hochschild boundary as
a Hodge dual of the absolute differential. This duality arises
from a ``quantum pairing'' which vanishes exactly in the commutative case.  
Nevertheless
we used this pairing to introduce a metric which generalises 
the classical metric, and which would
allow us to define noncommutative analogues of an action. In its 
present form, however, this approach is valid only in a flat background. 

Moreover, applying this formalism to projective modules, we gave a definition
of parallel transport in terms of connections, and we introduced a linear
homomorphism from the algebra to the endomorphisms of our right module, whose
algebraic curvature is related to the geometric curvature. This makes it possible to introduce 
{\itshape local} gauge invariant quantities; in the case of our 
quantum spacetime, this homomorphism maps the coordinates to the covariant 
coordinates  of \cite{Madore:2000en}.

\bigskip
\noindent{\bfseries Acknowledgements.} We are pleased to thank Joachim Cuntz 
and Jochen Zahn for helpful discussions and comments.

\footnotesize
\bibliographystyle{utphys}
\bibliography{bibfile}

\end{document}